\documentclass[11pt,reqno]{amsart}
\usepackage[utf8]{inputenc}
\usepackage[margin=1.18in]{geometry}
\usepackage{amsmath,amssymb,amsfonts,mathtools,mathrsfs}
\usepackage{enumitem}
\usepackage{microtype}
\usepackage[dvipsnames]{xcolor}
\usepackage[colorlinks=true,citecolor=MidnightBlue,linkcolor=BrickRed,
  urlcolor=RoyalBlue]{hyperref}
\usepackage[backend=bibtex,style=alphabetic,maxnames=99]{biblatex}
\AtEveryBibitem{%
  \iffieldundef{doi}{}{\clearfield{url}}%
  \iffieldequalstr{eprinttype}{arxiv}{\clearfield{url}}{}%
}

\numberwithin{equation}{section}
\newtheorem{theorem}{Theorem}[section]
\newtheorem{proposition}[theorem]{Proposition}
\newtheorem{lemma}[theorem]{Lemma}
\newtheorem{corollary}[theorem]{Corollary}
\theoremstyle{definition}
\newtheorem{definition}{Definition}[section]
\theoremstyle{remark}
\newtheorem{remark}[theorem]{Remark}
\newcommand{\EE}{\mathbb E}
\newcommand{\PP}{\mathbb P}
\newcommand{\FF}{\mathcal F}
\newcommand{\RP}{\mathbb {RP}^{1}}
\newcommand{\1}{\mathbf 1}
\newcommand{\norm}[1]{\left\lVert #1\right\rVert}

\title[Poisson statistics in one dimension]
{Poisson statistics for one-dimensional Anderson--Bernoulli
Schr\"odinger operators}
\author{Yi Han}
\thanks{Department of Mathematics, Massachusetts Institute of Technology,
Cambridge, MA.\newline Email: \texttt{hanyi16@mit.edu}.}
\thanks{A Lean~4 formalization is available online;
see Appendix~\ref{sec:formalization} for details.}
\hypersetup{
  pdftitle={Poisson statistics for one-dimensional Anderson--Bernoulli
  Schrodinger operators},
  pdfauthor={Yi Han}
}

\begin{document}

\begin{abstract}
For a one-dimensional discrete Schr\"odinger operator with bounded
nonconstant iid potential, we prove a unit-intensity Poisson limit
for local eigenvalues unfolded by the integrated density of states.
The limit holds at every energy in the interior of the almost-sure
spectrum. In particular, it holds for Bernoulli potentials at every
fixed nonzero coupling. The proof uses auxiliary phase functions
built from inverse stationary measures to obtain an exact IDS drift.
Projective contraction then yields exponential crossing times, and
a spatial subdivision argument gives independent Poisson counts in
disjoint energy windows. At energies where the IDS has a finite
positive derivative, the usual linear rescaling gives the same limit.
\end{abstract}

\maketitle
\enlargethispage{4pt}

\section{Introduction}
\label{sec:main}

We study local eigenvalue statistics for the one-dimensional
Anderson--Bernoulli model at fixed coupling. In a localized regime,
eigenfunctions with well separated localization centers have little
spatial overlap. This suggests that eigenvalues near a fixed energy,
arising from distant regions of a large sample, should behave almost
independently. Poisson statistics on the scale of the mean spacing
are therefore the natural prediction.

To prove this prediction, one must control the possibility that a
single region contributes several eigenvalues to a short energy
interval. For potential laws with a bounded density, Minami's
estimate provides this control; combined with localization, it
yields Poisson statistics \cite{Minami1996}.

The Bernoulli model poses a particular difficulty. Localization is
known at every nonzero coupling \cite{CarmonaKleinMartinelli1987},
but the atomic potential law prevents the classical proof of
Minami's estimate by spectral averaging. Bourgain obtained Poisson
statistics for certain sufficiently small Bernoulli couplings
satisfying arithmetic conditions \cite{Bourgain2013}.
Here we prove the Poisson limit at every fixed nonzero coupling
and every energy in the spectral interior. We normalize the
eigenvalues using the integrated density of states itself, so the
result also applies when this function has no density.

\subsection{Model and main result}
We state the result for a general bounded iid potential, which
includes the Bernoulli case.
Let $(v_j)_{j\in\mathbb Z}$ be iid real random variables with a
compactly supported law $\mu$ that is not a point mass. Write
$H_\omega$ for the operator on $\ell^2(\mathbb Z)$ given by
\[
 (H_\omega u)_j=u_{j-1}+v_{j-1}(\omega)u_j+u_{j+1}.
\]
Its restriction to sites $1,\ldots,n$ with Dirichlet conditions
$u_0=u_{n+1}=0$ is denoted by $H_n$, and its eigenvalues by
$\lambda_{1,n}<\cdots<\lambda_{n,n}$. The integrated density of
states (IDS) is
\begin{equation}\label{eq:ids-definition}
 N(t)=\lim_{L\to\infty}\frac1L\,
       \EE\#\{1\leq j\leq L:\lambda_{j,L}\leq t\}.
\end{equation}
Write $dN$ for the associated density-of-states measure.
Thus $N(t)$ is the limiting mean number of eigenvalues per site
below $t$.
The potential law remains fixed as $n\to\infty$.

Let $\Sigma$ be the almost-sure spectrum of $H_\omega$.
For the iid model, $\Sigma=\operatorname{supp}\mu+[-2,2]$
\cite[Theorem~2]{Stolz2011}. The IDS is continuous and
$\operatorname{supp}(dN)=\Sigma$
\cite[Section~5.2 and Corollary~5.14]{Teschl2000}.
Consequently, $N$ is strictly increasing on every interval
contained in $\operatorname{int}\Sigma$: if $u<v$ belong to such
an interval, then $N(v)-N(u)\geq dN((u,v))>0$.

An energy interval of IDS mass $a$ contains a limiting mean
fraction $a$ of the eigenvalues. The local scale for a sample of
$n$ sites is therefore an IDS mass of order $1/n$.
We use $N$ to unfold the eigenvalues: for
$E\in\operatorname{int}\Sigma$, define the point process
\[
 \xi_{n,E}=\sum_{j=1}^n
          \delta_{\,n[N(\lambda_{j,n})-N(E)]},
\]
where $\delta_z$ denotes unit point mass at $z$.
A window of unfolded length $w$ corresponds to an IDS increment
$w/n$. Let $\Pi$ denote the Poisson point process on $\mathbb R$
with Lebesgue intensity: its counts in disjoint bounded intervals
are independent Poisson variables with means equal to the interval
lengths. All point-process limits below are in distribution for the
vague topology.

Our main result gives this Poisson limit throughout the spectral
interior. We state it first in IDS coordinates, and then in the
usual linear coordinates at energies where $N$ has a finite
positive derivative.

\begin{samepage}
\begin{theorem}[Poisson local statistics]\label{thm:poisson}
Let the potentials be iid with a compactly supported law $\mu$
that is not a point mass. Fix $E\in\operatorname{int}\Sigma$.
\begin{enumerate}
\renewcommand{\labelenumi}{(\arabic{enumi})}
\item\label{part:ids-unfolding} The unfolded process satisfies
\[
 \xi_{n,E}\ \Longrightarrow\ \Pi.
\]
\item\label{part:linear-scaling} If $N$ is differentiable at $E$ and
$0<\rho(E):=N'(E)<\infty$, then
\[
 \sum_{j=1}^n\delta_{\,n\rho(E)(\lambda_{j,n}-E)}
       \ \Longrightarrow\ \Pi.
\]
\end{enumerate}
\end{theorem}
\end{samepage}

For a nondegenerate two-point law on $\{a,b\}$, $a<b$, the theorem
applies throughout the interior of $[a-2,a+2]\cup[b-2,b+2]$.
In particular, $v_j=\lambda\omega_j$ is allowed for every
$\lambda\ne0$ and every Bernoulli parameter
$\PP(\omega_j=1)=p\in(0,1)$. The two spectral bands form one
interval when $b-a\leq4$, including the touching case $b-a=4$.

Unfolding is useful even when the IDS has no density, as can occur
for Bernoulli potentials at large coupling \cite{MartinelliMicheli1987}.
The second conclusion requires only the derivative at $E$.

\subsection{Proof strategy}\label{sec:proof-strategy}
We obtain the spectral Poisson limit by studying first crossings
of a phase gap. For nearby energies
$E_-<E_+$, follow the lifted transfer-matrix phases $\theta_j$ and
$\phi_j$ through the same potential sample, starting both from zero.
Sturm oscillation relates the eigenvalue count in $(E_-,E_+]$ to
the number of crossings of $D_j=\theta_j-\phi_j$ through multiples
of $\pi$, up to a rounding error of one (Lemma~\ref{lem:sturm}). Set
\[
 T=\inf\{j\geq0:D_j\geq\pi\},\qquad
 \Delta N=N(E_+)-N(E_-).
\]
The main probabilistic result is the uniform exponential law
$\Delta N\,T\Longrightarrow\operatorname{Exp}(1)$ in
Theorem~\ref{thm:exit}. We will need this law for arbitrary
starting phases with a gap tending to zero, so that it can be
used again after each crossing.

To estimate $T$, we seek a function of the two phases with
constant conditional mean increment $\pi\Delta N$. The gap
itself has phase-dependent increments, so we first construct,
at each single energy, a phase function whose expected increment
is independent of the starting angle.

For this purpose, measure the separation of two angles by the
mass of the oriented arc between them. Equal expected increments
require this mass to be preserved on average under a forward
phase step. This is exactly inverse stationarity
(Section~\ref{sec:furstenberg}). Furstenberg's theorem provides
a non-atomic measure $\nu_t^-$ with this property
(Corollary~\ref{prop:furstenberg-used}), and we define
\[
 \mathcal H_t(x)=\pi\nu_t^-([0,x]),\quad 0\leq x\leq\pi,
 \qquad
 \mathcal H_t(x+\pi)=\mathcal H_t(x)+\pi.
\]
Writing $F_t(v,\theta)$ for one lifted phase step, inverse
stationarity and the Sturm formula give
(Lemma~\ref{lem:harmonic})
\[
 \int\bigl[\mathcal H_t(F_t(v,\theta))-\mathcal H_t(\theta)\bigr]
           \,d\mu(v)=\pi(1-N(t))
 \qquad\text{for every }\theta.
\]

The central construction is now to compare the two energies using
their own phase functions:
\[
 W_j=\mathcal H_{E_-}(\theta_j)-\mathcal H_{E_+}(\phi_j).
\]
Let $\FF_j$ record the potentials used before step $j$.
Since the next potential is independent of $\FF_j$, subtracting
the two single-energy identities gives
\[
 \begin{aligned}
 \EE[W_{j+1}-W_j\mid\FF_j]
 &=\pi(1-N(E_-))-\pi(1-N(E_+))\\
 &=\pi\Delta N,
 \end{aligned}
\]
(Proposition~\ref{prop:drift}).

For a vanishing initial gap, continuity and the vanishing
overshoot give $W_0\to0$ and $W_T\to\pi$ as $E_+-E_-\to0$.
The exact drift therefore yields
\[
 \pi\Delta N\,\EE T=\EE W_T-W_0,
 \qquad \Delta N\,\EE T\longrightarrow1
\]
(Lemma~\ref{lem:mean}).

To identify the law of $T$, we combine these mean exit-time
bounds with projective contraction from the positive Lyapunov
exponent (Lemma~\ref{lem:contraction}). This gives
$\Delta N\,\EE\sum_{j<T}|W_j|\to0$ (Lemma~\ref{lem:killed}).

The limits $\EE|W_T-\pi|\to0$ and
$\Delta N\,\EE\sum_{j<T}|W_j|\to0$ remove the two error terms
in the Laplace calculation (Section~\ref{sec:laplace-proof}).
For a vanishing initial gap, this yields
\[
 \EE e^{-s\Delta N T}\longrightarrow\frac1{1+s},
 \qquad s>0,
\]
uniformly in the starting phase.

After each crossing, subtracting the completed period leaves a
gap at most equal to the vanishing overshoot. The unused
potentials are independent of the past, and the uniform
hitting-time law applies to every entering phase. Thus each
conditional rescaled waiting time has the same exponential limit,
giving independent exponential waiting times in the limit and
hence Poisson crossing counts
(Corollaries~\ref{cor:restart} and~\ref{cor:counts}).

Different energy windows share the same potentials. To prove
their asymptotic independence, we divide the $n$ sites into $m$
spatial blocks. The gap left after removing completed periods
is small at block starts (Lemma~\ref{lem:residual-pinning}), so
the same conditional estimates apply there. For windows of IDS
mass of order $1/n$, the limiting single-crossing probabilities
have size $1/m$, while multiple crossings in the containing
energy interval have probability $O(m^{-2})$. The latter also bounds the chance
that one block contributes to two different windows. Taking
$n\to\infty$ first and then $m\to\infty$ gives the joint Poisson
law (Proposition~\ref{prop:joint-crossing-counts}).

Finally, the vanishing phase gap modulo $\pi$ can affect the Sturm rounding
only when the terminal phase is close to $\pi/2$ modulo $\pi$.
Convergence of that phase to a non-atomic law makes this event
negligible, so crossing counts and eigenvalue counts agree with
probability tending to one
(Lemma~\ref{lem:spectral-count-conversion}).

\subsection{Comparison with Bourgain and Hislop--Nakano}
\label{sec:comparison}
The classical Minami estimate cannot be applied directly to Bernoulli
potentials: its spectral-averaging proof requires regularity of the
single-site law, which an atomic law lacks
\cite[Section~1]{Bourgain2013}.
For Poisson statistics with Bernoulli disorder, the main earlier
results to compare with ours are Bourgain's theorem
\cite{Bourgain2013} and the work of Hislop and Nakano
\cite{HislopNakano2025}. Both use finite-volume counting estimates
to control the occurrence of two eigenvalues in a short interval.

\bigskip
\paragraph{Bourgain's argument.}
Bourgain proved Poisson statistics for H\"older-regular single-site
laws, which may be singular but are non-atomic. In that setting,
the H\"older concentration bound permits spectral averaging and a
Minami-type estimate
\cite[Section~1 and Proposition~5]{BourgainSingular2014}.
Bernoulli laws have atoms, so this step is unavailable and close
eigenvalues require a different argument
\cite[Section~1]{Bourgain2013}.

His Bernoulli theorem concerns symmetric two-point laws with coupling
in a specified class of small algebraic numbers, at energies in a
compact subinterval of $(-2,2)$
\cite[Propositions~1--2]{Bourgain2013}. The same arithmetic conditions,
together with small coupling and this energy restriction, allow him
to use group expansion to obtain both a $C^1$ density of states and
bounded densities of the Furstenberg measures
\cite{Bourgain2014Expansion}.

Bourgain's proof uses the bounded Furstenberg densities to obtain
a finite-volume Wegner estimate, and the $C^1$ density of states
to derive an asymptotic formula for the mean eigenvalue count.

Smoothness of the IDS alone still leaves the
problem of close eigenvalues. Bourgain handles near resonances by
turning two approximate eigenvectors into spectral events on disjoint
subintervals, where independence of the potential gives a product
bound. Together with localization, this gives the two-eigenvalue
bound needed for Poisson statistics
\cite[Sections~1 and~3--5]{Bourgain2013}.

\bigskip
\paragraph{What the hitting-time law changes.}
Our construction takes a different route: we determine the crossing
rate from the exact IDS drift of the auxiliary function $W_j$.
We therefore measure each energy window by its IDS mass and use
unfolded coordinates.
An unfolded window $(a,b]$ is an interval in the rescaled coordinate
$n[N(t)-N(E)]$. For large $n$, its corresponding energy window
$(E_-,E_+]$ has endpoints determined by
\[
 N(E_-)=N(E)+\frac an,\qquad N(E_+)=N(E)+\frac bn.
\]
Its IDS mass is therefore $\Delta N=(b-a)/n$.
Because the drift identity \eqref{eq:drift} uses this mass exactly,
identifying the intensity requires no local linear approximation
of $N$. Our argument neither assumes nor proves that the
density-of-states measure $dN$ has a density. Whenever $N'(E)$
exists and is finite and positive, a change of scale yields the
linear normalization in Theorem~\ref{thm:poisson}. Existence of
this derivative is not an assumption of the unfolded Poisson limit.

The required control of multiple counts also becomes a consequence
of the hitting-time theorem. Corollaries~\ref{cor:restart}
and~\ref{cor:counts} yield the two-arrival bound needed for spatial
subdivision. This control is obtained for phase crossings before
converting them to eigenvalue counts. The spectral Poisson argument
therefore needs no additional near-resonance estimate for a pair of
eigenfunctions.

\bigskip
\paragraph{Hislop--Nakano: unfolding and localization.}
\label{sec:hn-comparison}
Hislop and Nakano treat random polymer models, including ordinary
Anderson--Bernoulli operators. Away from critical energies, they
combine IDS unfolding with multiscale localization and independent
subsystems. Their Minami-type estimate separates eigenfunctions with
distant localization centers, controlled by Wegner estimates and
independence, from nearby pairs, controlled by an adaptation of
Bourgain's lemma
\cite[Sections~3.1 and~3.3--3.7; Appendix~D]{HislopNakano2025}.

Their Poisson theorem assumes the following local H\"older regularity
conditions on the IDS and its inverse: for some
$\rho_1,\rho_2\in(0,1]$ with $\rho_1\rho_2>2/3$,
\[
 \begin{aligned}
 |N(t)-N(t')|&\leq C|t-t'|^{\rho_1},\\
 |t-t'|&\leq C|N(t)-N(t')|^{\rho_2},
 \end{aligned}
 \qquad\text{for all }t,t'\text{ near }E.
\]
The inverse bound turns an unfolded interval of size $n^{-1}$ into
an energy interval of width $O(n^{-\rho_2})$. The product condition
allows a polynomial choice of subsystem size for which both multiple
eigenvalues and losses near subsystem boundaries are negligible
\cite[Assumption~3.2, Theorem~3.3, Remark~3.4 and~(3.29)--(3.33)]{HislopNakano2025}.
These are additional quantitative hypotheses; the usual H\"older
upper bound for $N$ does not verify them. The required inverse-IDS
regularity is not established there, and the authors state that they
know no example satisfying these conditions, leaving their
verification in concrete models open
\cite[Remark~3.4]{HislopNakano2025}. Their results also cover
correlated polymers and, under further hypotheses, clock statistics
at critical energies \cite[Theorem~2.18]{HislopNakano2025}, beyond
the iid setting considered here.

Both our proof and theirs use unfolding. The difference is that Hislop--Nakano
also need quantitative regularity to make the errors from spectral
subdivision small on the unfolded scale. Our error estimates are
already on the IDS time scale (Section~\ref{sec:proof-strategy});
they require only that the energy window shrinks, with no prescribed
rate. Thus the inverse-IDS bound and the restrictions on the
H\"older exponents do not enter our proof.

\subsection{Other related work}\label{sec:related}
The classical Poisson theory combines localization with estimates
on the number of eigenvalues in a short interval. Molchanov
\cite{Molchanov1981} proved a continuum Poisson limit in one
dimension. For discrete Anderson models, the Wegner and Minami
estimates control one and multiple eigenvalues, and localization
reduces the problem to almost independent boxes
\cite{Wegner1981,Minami1996}. Combes, Germinet, and Klein
\cite{CombesGerminetKlein2009} extended the counting estimates to
non-atomic laws, including singular laws. A standard finite-rank
comparison for spectral counts is recorded in
\cite[Lemma~2.2]{Bai1999}.

Poisson statistics have also been obtained for one-dimensional
random band matrices \cite{FanZou2026}. For lattice Bernoulli models,
localization has been proven near the bottom of the spectrum in
two, three and higher dimensions
\cite{DingSmart2020,LiZhang2022,LiLiuZhang2026}.
For symmetric two-point laws, the corresponding result near the
top follows by spectral reflection.

To study local statistics on the IDS scale, Germinet and Klopp
\cite{GerminetKlopp2014} developed an unfolded theory under
localization, counting estimates and quantitative non-flatness
conditions on the IDS. For bounded ergodic potentials, Craig and
Simon \cite{CraigSimon1983} prove the upper bound
$N(t+h)-N(t)\leq C/|\log h|$ for $0<h<1/2$.
For bounded nonconstant iid potentials in one dimension, this
improves locally to $N(t+h)-N(t)\leq Ch^\alpha$ for some
$\alpha>0$ \cite[Theorem~3]{LePage1989}; see also
\cite{CarmonaKleinMartinelli1987}. These upper bounds do not give
the lower bounds on IDS increments required for inverse-IDS control.

Furstenberg's theorem and Le Page's theory provide the stationary
laws and projective contraction for our transfer maps
\cite{Furstenberg1963,LePage82,LePage1989}. The inverse law
$\nu_t^-$ in Corollary~\ref{prop:furstenberg-used} defines
$\mathcal H_t(x)=\pi\nu_t^-([0,x])$ for $0\leq x\leq\pi$
(Section~\ref{sec:furstenberg}). Its increments are therefore
$\pi$ times the masses of oriented projective arcs. This arc-mass
construction is standard in random circle dynamics
\cite{Navas2011,GKGV2014}; see in particular
\cite[Proposition~1]{GelfertStenflo2017}.
In our setting, inverse stationarity makes the mean increment
independent of the starting phase (Lemma~\ref{lem:harmonic}).
Taking the difference along the coupled phase recursions gives $W_j$ in
Definition~\ref{def:auxiliary-gap}, with the exact IDS drift of
Proposition~\ref{prop:drift}. A formula expressing rotation-number
increments through stationary arc masses is given in
\cite[Theorem~2.4]{DuarteGorodetskiKleptsyn2025}.

For one-dimensional continuum Schr\"odinger operators with decaying
potentials, Kotani and Nakano \cite{KN17} prove Poisson statistics
by studying phase evolution and comparing it with diffusions.
Allez and Dumaz \cite{AD14} use related phase diffusions to prove
the small-$\beta$ Poisson limit of $\operatorname{Sine}_\beta$.
Critical bulk and edge limits for one-dimensional operators with
decaying or vanishing potentials are studied in
\cite{KVV12,HanCritical2026,Han23}.
In the present paper, the potential law stays fixed as the volume
grows.

\section{Phase coordinates and the IDS}
\label{sec:coordinates}

This section develops the phase description used in the proof.
Transfer matrices turn the eigenvalue equation into an angular
recursion, and Sturm oscillation recovers the eigenvalue count
from its terminal phase. We then compare the phases at two
energies and construct an auxiliary function of their phases
with an exact IDS drift. This gives the starting point for the
crossing-time estimates in the next section.

Throughout this section, $E\in\operatorname{int}\Sigma$ is the
fixed reference energy, and $E'$ denotes a variable energy.

\subsection{Transfer matrices and canonical lifts}
We begin by rewriting the eigenvalue equation as a first-order
recursion for pairs of consecutive solution values. This will
allow us to follow the direction of the solution vector and use
its phase to count eigenvalues.

For a trial energy $E'$, consider the solution of the difference
equation $H_\omega u=E'u$ with $u_0=0$ and $u_1=1$. These
conditions impose the left Dirichlet boundary and fix the
normalization; $E'$ is an eigenvalue of $H_n$ precisely when the
remaining boundary condition $u_{n+1}=0$ holds. The recursion
$u_{j+2}=(E'-v_j)u_{j+1}-u_j$ takes the matrix form
\[
 \begin{pmatrix}u_{j+2}\\u_{j+1}\end{pmatrix}
 =A_{E'}(v_j)\begin{pmatrix}u_{j+1}\\u_j\end{pmatrix},
 \qquad j\geq0,\qquad
 A_{E'}(v)=\begin{pmatrix}E'-v&-1\\1&0\end{pmatrix}.
\]
For $j\geq1$, write the transfer product in its order of action as
\[
 M_{E',j}=A_{E'}(v_{j-1})\cdots A_{E'}(v_0),
 \qquad M_{E',0}=\mathrm{Id}.
\]
Thus $M_{E',j}(1,0)^{\mathsf T}=(u_{j+1},u_j)^{\mathsf T}$.
The right Dirichlet condition concerns the direction of this
vector. We describe directions by the phase angle $\theta$, using
\[
 \mathbf u(\theta):=(\cos\theta,\sin\theta)^{\mathsf T},
 \qquad \theta\in\mathbb R.
\]
A direction determines $\theta$ only modulo $\pi$, whereas
eigenvalue counting requires the accumulated phase. We therefore
choose a real-valued lift of the angular action of each transfer
matrix, with the following normalization.

\begin{samepage}
\begin{definition}[Canonical phase lift]\label{def:phase-lift}
For $E',v\in\mathbb R$, the canonical phase lift
$F_{E'}(v,\cdot):\mathbb R\to\mathbb R$ is the unique continuous
function satisfying
\[
 \mathbf u\bigl(F_{E'}(v,\theta)\bigr)
   =\frac{A_{E'}(v)\mathbf u(\theta)}
          {\norm{A_{E'}(v)\mathbf u(\theta)}}
 \qquad(\theta\in\mathbb R)
\]
with the normalization and periodicity conditions
\begin{equation}\label{eq:lift}
 F_{E'}(v,\pi/2)=\pi,\qquad
 F_{E'}(v,\theta+\pi)=F_{E'}(v,\theta)+\pi.
\end{equation}
\end{definition}
\end{samepage}

To see explicitly that this lift is well defined, write $a=E'-v$.
For $-\pi/2<\theta<\pi/2$, the second component of
$A_{E'}(v)\mathbf u(\theta)$ is positive, so its argument lies
in $(0,\pi)$ and is given by
\[
 F_{E'}(v,\theta)=\frac{\pi}{2}
                  -\arctan(a-\tan\theta),
 \qquad -\frac{\pi}{2}<\theta<\frac{\pi}{2},
\]
where $\arctan$ takes values in $(-\pi/2,\pi/2)$.
The endpoint limits are $0$ and $\pi$, respectively. Extending
continuously to the endpoints and then by the second identity in
\eqref{eq:lift} constructs the lift on all of $\mathbb R$.
Continuity and the prescribed value at $\pi/2$ fix its branch
uniquely.

We will compare phases at different initial angles and nearby
energies, so we record how this lift varies in both arguments.
For a differentiable nonzero planar vector
$z(t)=R(t)(\cos\alpha(t),\sin\alpha(t))^{\mathsf T}$,
where $R(t)=\norm{z(t)}>0$ and $\alpha$ is a continuous argument,
differentiation gives
\[
 \det(z,z')=R^2\alpha',\qquad
 \alpha'=\frac{\det(z,z')}{\norm{z}^2}.
\]
In the present case,
\[
 z=A_{E'}(v)\mathbf u(\theta)
   =\begin{pmatrix}(E'-v)\cos\theta-\sin\theta\\\cos\theta\end{pmatrix},
 \qquad \partial_{E'}z=(\cos\theta,0)^{\mathsf T}.
\]
Since $\det A_{E'}(v)=1$ and
$\det(\mathbf u(\theta),\mathbf u'(\theta))=1$, we have
\[
 \det(z,\partial_\theta z)=1,\qquad
 \det(z,\partial_{E'}z)=-\cos^2\theta.
\]
Substituting into the argument-derivative formula gives
\begin{equation}\label{eq:derivatives}
 \partial_\theta F_{E'}(v,\theta)
   =\frac1{\norm{A_{E'}(v)\mathbf u(\theta)}^2}>0,
 \qquad
 \partial_{E'} F_{E'}(v,\theta)
   =-\frac{\cos^2\theta}{\norm{A_{E'}(v)\mathbf u(\theta)}^2}.
\end{equation}
The denominator is
$((E'-v)\cos\theta-\sin\theta)^2+\cos^2\theta$,
so the energy derivative lies in $[-1,0]$. Hence, for $\delta\geq0$,
\begin{equation}\label{eq:energyclose}
 0\leq F_{E'}(v,x)-F_{E'+\delta}(v,x)\leq\delta.
\end{equation}
Choose $M<\infty$ such that $|v_j|\leq M$ almost surely.
On a fixed compact energy interval $I$, all matrix norms and inverse
norms are bounded uniformly for $v\in[-M,M]$. The phase maps and
their inverses therefore have uniformly bounded Lipschitz constants.

\subsection{Sturm oscillation}
We now relate the accumulated phase to the eigenvalue count.
The right Dirichlet condition is met whenever the solution vector
is vertical, so eigenvalues occur when the terminal phase crosses
$\pi/2$ modulo $\pi$. The discrete Sturm theorem counts these
crossings \cite[Theorems~4.5, 4.7 and equation~(4.40)]{Teschl2000}.
We state the formula with the lift and endpoint conventions used
throughout the proof. Define the Dirichlet phases by
\begin{equation}\label{eq:dirichlet-phase}
 \theta_0(E')=0,\qquad
 \theta_{j+1}(E')=F_{E'}(v_j,\theta_j(E')).
\end{equation}
Then $\theta_n(E')$ is the continuous lifted argument of
$(u_{n+1},u_n)$.

The next lemma lets us count eigenvalues below $E'$ by computing
$\theta_n(E')$ through the phase recursion and taking an integer
part as in \eqref{eq:sturm}. We therefore do not need to solve
$\theta_n(t)=\pi/2\bmod\pi$ for each eigenvalue $t$.

\begin{samepage}
\begin{lemma}[Sturm phase identities]\label{lem:sturm}
For every realization and energy $E'$,
\begin{equation}\label{eq:sturm}
 \#\{1\leq j\leq n:\lambda_{j,n}<E'\}
    =n-\left\lfloor\frac{\theta_n(E')+\pi/2}{\pi}\right\rfloor.
\end{equation}
For $E'<E''$,
\begin{equation}\label{eq:sturm-interval}
 \#\{1\leq j\leq n:\lambda_{j,n}\in(E',E'']\}
 =\left\lceil\frac{\theta_n(E')-\pi/2}{\pi}\right\rceil
  -\left\lceil\frac{\theta_n(E'')-\pi/2}{\pi}\right\rceil.
\end{equation}
\end{lemma}
\end{samepage}

\begin{proof}
The Dirichlet eigenvalue condition is $u_{n+1}(E')=0$, or
$\theta_n(E')=\pi/2\bmod\pi$. Differentiating the phase recursion
using \eqref{eq:derivatives} gives
\[
 \frac{d\theta_{j+1}}{dE'}
 =\frac{d\theta_j/dE'-\cos^2\theta_j}
        {\norm{A_{E'}(v_j)\mathbf u(\theta_j)}^2}.
\]
The initial phase is constant, so $d\theta_0/dE'=0$, while
$d\theta_1/dE'=-1/[(E'-v_0)^2+1]<0$.
In the displayed recursion the denominator is positive, and a
negative input derivative makes the numerator negative. Induction
therefore shows that $\theta_n$ is strictly decreasing for every
$n\geq1$.
As $E'$ increases from $-\infty$ to $+\infty$, $\theta_n(E')$
decreases from $n\pi$ to zero. Its successive half-integer crossings
give \eqref{eq:sturm}. The non-strict count is
$n-\lceil(\theta_n(E')-\pi/2)/\pi\rceil$; subtracting it at the two
endpoints gives \eqref{eq:sturm-interval}.
\end{proof}

\subsection{An auxiliary function for the phase gap}
\label{sec:furstenberg}
\label{sec:cdf}
The Sturm formula counts eigenvalues in an energy interval by
comparing the phases at its two endpoints. We therefore consider
two nearby energies $E_-<E_+$ in $\operatorname{int}\Sigma$
and run both phase recursions with the
same potentials, starting from deterministic phases $\theta_0$
and $\phi_0$:
\begin{equation}\label{eq:pair}
 \theta_{j+1}=F_{E_-}(v_j,\theta_j),\qquad
 \phi_{j+1}=F_{E_+}(v_j,\phi_j),\qquad
 D_j=\theta_j-\phi_j.
\end{equation}
Write $\Delta N=N(E_+)-N(E_-)>0$. Our goal is to estimate
when $D_j$ first reaches $\pi$. Its conditional mean increment
depends on the current phases, so a stopping argument would
require control of this dependence. We therefore construct an auxiliary
function $W_j$ of the two phases whose conditional mean increment
is known. If its values at the start and at the crossing are
also controlled, stopping its drift will give an estimate for
the crossing time of $D_j$.

To construct $W_j$, we first seek, for each energy $t$, a
function $\mathcal H_t$ of a single phase whose expected
increment does not depend on the starting angle. The inverse
stationary measure provides this function, as follows.

For a nonzero vector $u$, write $[u]\in\RP$ for its direction.
Identify $\RP$ with $\mathbb R/\pi\mathbb Z$ by
$[\mathbf u(\theta)]\leftrightarrow\theta\bmod\pi$, and let
$T_{t,v}$ be the circle map induced by $F_t(v,\cdot)$. Thus
\[
 T_{t,v}[u]=[A_t(v)u],\qquad
 T_{t,v}^{-1}[u]=[A_t(v)^{-1}u].
\]
If the transformed phases at two starting angles have the same
expected increment, their separation has zero expected change.
Suppose we measure their separation
by the mass of the oriented arc $J$ between them under a
probability measure $\nu$. A forward step sends $J$ to $T_{t,v}J$,
so the required condition is
\[
 \int\nu(T_{t,v}J)\,d\mu(v)=\nu(J).
\]
Since $(T_{t,v}^{-1})_*\nu(J)=\nu(T_{t,v}J)$, this is exactly
stationarity for the inverse maps. This explains why we use an
inverse stationary measure even though the phases evolve forward.
At energy $E'$, we denote this measure by $\nu_{E'}^-$ and write
\begin{equation}\label{eq:inverse}
 \nu_{E'}^-=\int (T_{E',v}^{-1})_*\nu_{E'}^-\,d\mu(v).
\end{equation}
Forward stationarity is defined in the same way, with $T_{t,v}$
in place of $T_{t,v}^{-1}$; its law is denoted by $\nu_t^+$.
Furstenberg's theorem gives both laws for our matrices
\cite{Furstenberg1963}; see also
\cite[Section~3.1]{HochmanSolomyak2017}.

\begin{samepage}
\begin{corollary}[Furstenberg]\label{prop:furstenberg-used}
For every real $E'$, the forward and inverse actions have unique
non-atomic stationary probabilities $\nu_{E'}^+$ and $\nu_{E'}^-$.
\end{corollary}
\end{samepage}

\begin{proof}
Two distinct potential values give a nontrivial shear in the
transfer group, and the transfers do not preserve its fixed
line; Furstenberg's theorem therefore applies to both laws.
\end{proof}

We now define the auxiliary phase function $\mathcal H_t$ as
the lifted cumulative distribution function of $\nu_t^-$, with
$\mathcal H_t(0)=0$ and
\begin{equation}\label{eq:phase-function}
 \mathcal H_t(x)=\pi\nu_t^-([0,x]),\qquad 0\leq x\leq\pi,
 \qquad
 \mathcal H_t(x+m\pi)=\mathcal H_t(x)+m\pi.
\end{equation}
Non-atomicity makes $\mathcal H_t$ continuous and makes the
endpoint convention immaterial. It is nondecreasing and may
have flat portions; no invertibility is required. The difference
$\mathcal H_t(x)-x$ is bounded and $\pi$-periodic.
We use $F$ for real lifts
and $T$ for maps of projective points and arcs.

The arc identity above makes the mean increment of
$\mathcal H_t$ independent of the starting angle. The IDS
determines the value of this mean.

\begin{samepage}
\begin{lemma}[The mean phase increment]\label{lem:harmonic}
For every initial angle $\theta$,
\begin{equation}\label{eq:exactone}
 \int[\mathcal H_{E'}(F_{E'}(v,\theta))
                  -\mathcal H_{E'}(\theta)]\,d\mu(v)
       =\pi(1-N(E')).
\end{equation}
For every $x\geq y$, the corresponding arc masses satisfy
\begin{equation}\label{eq:exactgap}
 \int[\mathcal H_{E'}(F_{E'}(v,x))
                  -\mathcal H_{E'}(F_{E'}(v,y))]\,d\mu(v)
       =\mathcal H_{E'}(x)-\mathcal H_{E'}(y).
\end{equation}
\end{lemma}
\end{samepage}

\begin{proof}
For $0\leq x-y<\pi$, let $J$ be the positively oriented arc from
$y$ to $x$. Since $F_{E'}(v,\cdot)$ is an increasing degree-one
lift,
\[
 \mathcal H_{E'}(F_{E'}(v,x))
      -\mathcal H_{E'}(F_{E'}(v,y))
   =\pi\nu_{E'}^-(T_{E',v}J).
\]
The lifted difference includes the wrapped part when an arc
crosses the cut. Inverse stationarity gives
\[
 \int\nu_{E'}^-(T_{E',v}J)\,d\mu(v)
 =\int(T_{E',v}^{-1})_*\nu_{E'}^-(J)\,d\mu(v)
 =\nu_{E'}^-(J),
\]
which proves \eqref{eq:exactgap} for $0\leq x-y<\pi$.
The shift identities \eqref{eq:lift} and \eqref{eq:phase-function}
extend it to all $x\geq y$.

By \eqref{eq:exactgap}, the average increment on the left side
of \eqref{eq:exactone} is independent of $\theta$; call it $c$.
To identify $c$, let $\psi_j$ be the phase at energy $E'$ with
$\psi_0=0$. Summing the mean increments gives
$\EE\mathcal H_{E'}(\psi_n)=nc$.
Since $\mathcal H_{E'}(x)-x$ is bounded,
$\EE\psi_n/n\to c$. The Sturm identity
\eqref{eq:sturm} therefore gives $N(E')=1-c/\pi$.
Thus $c=\pi(1-N(E'))$, proving \eqref{eq:exactone}.
\end{proof}

Lemma~\ref{lem:harmonic} gives a constant mean increment for
$\mathcal H_{E_-}(\theta_j)$ and for
$\mathcal H_{E_+}(\phi_j)$. Their difference is therefore the
auxiliary function we need.

\begin{samepage}
\begin{definition}[The auxiliary function $W_j$]\label{def:auxiliary-gap}
For the coupled phases in \eqref{eq:pair}, let
$\FF_j=\sigma(\theta_0,\phi_0,v_0,\ldots,v_{j-1})$.
Define
\begin{equation}\label{eq:W}
 W_j=\mathcal H_{E_-}(\theta_j)-\mathcal H_{E_+}(\phi_j).
\end{equation}
\end{definition}
\end{samepage}

\begin{remark}[Why we use $W_j$]\label{rem:auxiliary-gap}
We introduce $W_j$ as an auxiliary function of the two phases
to study the gap $D_j$. Its advantage is that its mean increase
at each step is the known constant $\pi\Delta N$, while that of
$D_j$ depends on the current phases. We can therefore use the
drift identity for $W_j$ at stopping times defined using $D_j$.
For nearby energies, the values of $W_j$ are controlled when
$D_j$ is near $0$ or $\pi$, so this identity gives estimates
for the time it takes $D_j$ to cross a period. The details are
given in Lemma~\ref{lem:mean}.
\end{remark}

The same potential is used in both recursions, so the two phases
are not independent. The mean-increment formula still applies
to each term separately, and subtracting gives the desired drift.

\begin{samepage}
\begin{proposition}[Exact IDS drift]\label{prop:drift}
For the coupled phases in \eqref{eq:pair} and the transformed
gap in \eqref{eq:W},
\begin{equation}\label{eq:drift}
 \EE[W_{j+1}-W_j\mid\FF_j]=\pi \Delta N
      =\pi[N(E_+)-N(E_-)].
\end{equation}
\end{proposition}
\end{samepage}

\begin{proof}
Since $v_j$ is independent of $\FF_j$, \eqref{eq:exactone} gives
\[
 \begin{aligned}
 \EE[W_{j+1}-W_j\mid\FF_j]
 &=\pi(1-N(E_-))-\pi(1-N(E_+))\\
 &=\pi \Delta N.
 \end{aligned}
\]
\end{proof}

The function $\mathcal H_t$ also varies continuously with $t$.
To compare nearby energies, we use the uniform continuity
estimates below.

\begin{samepage}
\begin{lemma}[Continuity of the stationary laws and phase functions]
\label{lem:continuity}
The stationary laws $\nu_t^+$ and $\nu_t^-$ vary weakly
continuously with $t$. The normalized lifts satisfy
$\|\mathcal H_t-\mathcal H_{t'}\|_\infty\to0$ as $t'\to t$.
Their differences are $\pi$-periodic, so this norm can be taken
over $[0,\pi]$ or over $\mathbb R$.

On a compact energy interval $I$, define
\[
 \omega_{\mathrm{en}}(\delta)
 =\sup_{\substack{t,t'\in I\\|t-t'|\leq\delta}}
       \|\mathcal H_t-\mathcal H_{t'}\|_\infty,
 \qquad
 \omega_{\mathrm{ph}}(h)
 =\sup_{t\in I}\sup_{|x-y|\leq h}
       |\mathcal H_t(x)-\mathcal H_t(y)|.
\]
Here $\mathrm{en}$ stands for energy and $\mathrm{ph}$ for phase.
Both moduli tend to zero with their arguments.
\end{lemma}
\end{samepage}

\begin{proof}
For either sign, stationarity means that, for every $g\in C(\RP)$,
\[
 \int g(x)\,d\nu_{t_m}^{\pm}(x)
 =\int_{\RP}\int_{\mathbb R}
       g(T_{t_m,v}^{\pm1}x)\,d\mu(v)\,d\nu_{t_m}^{\pm}(x).
\]
If $t_m\to t$, continuity of the matrix action and compactness
of the potential support and the projective line give uniform
convergence of $g(T_{t_m,v}^{\pm1}x)$ to $g(T_{t,v}^{\pm1}x)$.
Any weak subsequential limit of $\nu_{t_m}^{\pm}$ therefore
satisfies the stationary equation at $t$. Uniqueness identifies
it with $\nu_t^{\pm}$, proving weak continuity.

For the inverse laws, non-atomicity of the limiting measure
upgrades weak convergence to uniform convergence of the
normalized distribution functions. Indeed, for each $\varepsilon>0$
choose a finite partition of $[0,\pi]$ on which each increment
of $\mathcal H_t$ is less than $\varepsilon$. Weak convergence
gives convergence at these partition points, and monotonicity
bounds the error between them by the endpoint errors plus
$\varepsilon$. On compact $I$, the resulting
continuous family of continuous lifts is uniformly continuous
in energy and equicontinuous in phase. This gives both moduli.
\end{proof}

\section{First hitting times and the Poisson limit}
\label{sec:exit}

The purpose of this section is to prove the uniform first
hitting-time law in Theorem~\ref{thm:exit}. The key conceptual
reduction is that this law leads to the spectral Poisson limit
of Theorem~\ref{thm:poisson}.
Recall the coupled phases in \eqref{eq:pair}, with
$D_j=\theta_j-\phi_j$ and $K_j=\lfloor D_j/\pi\rfloor$.
For the Dirichlet initial phases, the Sturm identity
\eqref{eq:sturm-interval} of Lemma~\ref{lem:sturm} gives
\[
 0\leq\#\{k:\lambda_{k,n}\in(E_-,E_+]\}-K_n\leq1.
\]
We therefore study the times at which $D_j$ completes a period.
Theorem~\ref{thm:exit} asserts that
\[
 \begin{gathered}
 T_m=\inf\{j\geq0:D_j\geq m\pi\},\qquad
 \Delta N=N(E_+)-N(E_-),\\
 \Delta N\,T_1\Longrightarrow\operatorname{Exp}(1)
 \qquad(E_+-E_-\to0),
 \end{gathered}
\]
uniformly in the initial phase and in an initial gap tending to
zero. The uniformity in Theorem~\ref{thm:exit} lets us restart
after each crossing, using the fresh iid potentials and the
vanishing overshoot. With $T_0=0$, the resulting limits are,
for each fixed $M$,
\[
 \begin{gathered}
 \Delta N\bigl(T_1,T_2-T_1,\ldots,T_M-T_{M-1}\bigr)
       \Longrightarrow(X_1,\ldots,X_M),\qquad
       X_i\overset{\mathrm{iid}}{\sim}\operatorname{Exp}(1),\\
 K_{\lfloor t/\Delta N\rfloor}
       \Longrightarrow\operatorname{Poisson}(t),\qquad t\geq0.
 \end{gathered}
\]
We state these consequences of Theorem~\ref{thm:exit} as
Corollaries~\ref{cor:restart} and~\ref{cor:counts}.
Section~\ref{sec:pointprocess} then proves independence across
energy windows and removes the one-count error above, giving
$\xi_{n,E}\Longrightarrow\Pi$.

\begin{samepage}
We prove Theorem~\ref{thm:exit} in three main parts.
\begin{enumerate}[label=(\roman*)]
\item \emph{Positive Lyapunov exponent and contraction.}
This is the main probabilistic input. Positive growth of transfer
products gives projective contraction, which controls how long
the gap stays away from the two ends of a period
(Section~\ref{sec:contraction}).
\item \emph{The auxiliary function $W_j$.}
Definition~\ref{def:auxiliary-gap} and
Proposition~\ref{prop:drift} give the exact drift
$\EE[W_{j+1}-W_j\mid\FF_j]=\pi\Delta N$.
Stopping this identity gives bounds on the mean hitting time and on the time
remaining near the upper end. Together with contraction, these
bounds control the total size of $W_j$ before the crossing
(Section~\ref{sec:mean-occupation}).
\item \emph{The Laplace transform.}
We apply the drift identity to a discounted affine function of
$W_j$. The remaining estimates make the error terms vanish,
giving $\EE e^{-s\Delta N T_1}\to(1+s)^{-1}$
(Section~\ref{sec:laplace-proof}).
\end{enumerate}
\end{samepage}

Fix $E\in\operatorname{int}\Sigma$. For an ordered energy pair
$E_-<E_+$ near $E$, write
\[
 \delta=E_+-E_-,\qquad
 \Delta N=N(E_+)-N(E_-)>0.
\]
The energy interval $I$ will be chosen in
Lemma~\ref{lem:contraction}. All constants and moduli below are
uniform for energy pairs in $I$ and may depend on $\mu$ and $E$.

\subsection{The first hitting-time problem}
Recall the coupled phases $\theta_j$ and $\phi_j$ from
\eqref{eq:pair}, driven by the same potentials at $E_-$ and $E_+$.
Their lifted gap is $D_j=\theta_j-\phi_j$.
Start from deterministic phases with
$0\leq D_0=\theta_0-\phi_0<\pi$. The first hitting time of one
full period is
\[
 T=\inf\{j\geq0:D_j\geq\pi\}.
\]
For successive crossings and their counts, also set
\[
 T_m=\inf\{j\geq0:D_j\geq m\pi\},\qquad
 K_j=\lfloor D_j/\pi\rfloor.
\]
We regard the spatial index $j$ as the time parameter of this
process. Thus $T_1=T$. The time $T_m$ records the first crossing
of $m\pi$, and $K_j$ records the current number of complete periods.
We take $\inf\varnothing=\infty$ and use $W_j$ and $\FF_j$ from
Definition~\ref{def:auxiliary-gap}. The conditional drift of
$W_j$ is $\pi\Delta N$ by Proposition~\ref{prop:drift}.

We first state the exponential law for a vanishing initial gap.
The more general Laplace-transform profile will be proved at
the same time.

\begin{samepage}
\begin{theorem}[Uniform exponential crossing law]\label{thm:exit}
There exists a compact interval $I\subset\operatorname{int}\Sigma$
with $E$ in its interior such that, for every $h_\delta\geq0$
tending to zero,
\begin{equation}\label{eq:exp}
 \Delta N\,T\ \Longrightarrow\ \operatorname{Exp}(1)
 \qquad(\delta\downarrow0),
\end{equation}
uniformly over $E_-<E_+$ in $I$, the initial phase $\theta_0$,
and $0\leq D_0\leq h_\delta$.

More generally, for every fixed $s>0$,
\begin{equation}\label{eq:affineprofile}
 \sup_{\theta_0\in\mathbb R,\,0\leq D_0<\pi}
 \left|\EE e^{-s\Delta N\,T}
       -\frac{1+(s/\pi)W_0}{1+s}\right|\longrightarrow0,
\end{equation}
uniformly over the same energy pairs as $\delta\downarrow0$.
\end{theorem}
\end{samepage}

\pagebreak[0]
\begin{samepage}
\begin{remark}[The two errors in the Laplace argument]
To see what the estimates below must accomplish, fix $s>0$ and
put $\rho=(1+s\Delta N)^{-1}$. Equation~\eqref{eq:drift} gives
$\EE[W_{j+1}\mid\FF_j]=W_j+\pi\Delta N$, so we can compute
the conditional mean increment of $\rho^j(1+(s/\pi)W_j)$.
Summing these increments up to $T$ and separating the exit error
$W_T-\pi$ gives
\begin{equation}\label{eq:laplace-errors}
 \begin{aligned}
 &\left|(1+s)\EE\rho^T-1-\frac{s}{\pi}W_0\right|\\
 &\qquad\leq\frac{s}{\pi}\EE|W_T-\pi|
       +\frac{s^2}{\pi}\Delta N\,\EE\sum_{j<T}|W_j|.
 \end{aligned}
\end{equation}
The stopping is justified once the mean bound of
Lemma~\ref{lem:mean} is available; the calculation is given in
Section~\ref{sec:laplace-proof}. Thus the Laplace argument has
two errors to remove: the error in the value of $W$ at exit,
and the accumulated size of $W$ before exit. The mean bound also
lets us replace $\rho^T$ by $e^{-s\Delta N T}$.

The first error is controlled by the overshoot and continuity.
For the second, Lyapunov contraction controls time spent in the
middle of the phase interval, while the drift of $W$ controls
the time remaining near $\pi$. These estimates are uniform in
the starting state and require no quantitative comparison
between $\delta$ and $\Delta N$.
\end{remark}
\end{samepage}

\subsection{Positive Lyapunov exponent and contraction}
\label{sec:contraction}
The main probabilistic input is positivity of the Lyapunov
exponent. It gives contraction of two projective directions
under the same random transfer product. For nearby energies,
this will bound the time the lifted gap spends in the middle
of $[0,\pi]$, and hence its contribution to the accumulated
error in \eqref{eq:laplace-errors}.

We first work at a single energy. On the projective line, use
the distance
\begin{equation}\label{eq:projective-distance}
 d([u],[v])=\frac{|\det(u,v)|}{\norm{u}\norm{v}}
            =|\sin(\theta-\phi)|.
\end{equation}
Here $u,v\ne0$ represent the two projective points, and
$\theta,\phi$ are their angles modulo $\pi$.
This is exactly Le Page's distance
\cite[Section~II, equation~(6)]{LePage1989}: in dimension two,
the exterior-product norm is $|\det(u,v)|$.
This distance induces the usual topology on $\RP$ and is equivalent
to angular distance modulo $\pi$.
For unit representatives $x,y$ and $A\in\mathrm{SL}(2,\mathbb R)$,
the identity $\det A=1$ gives
\[
 \det(Ax,Ay)=\det A\,\det(x,y)=\det(x,y).
\]
Using this in the definition of $d$, and observing that
$d(x,y)=|\det(x,y)|$ for unit $x,y$, we obtain
\begin{equation}\label{eq:sl2distance}
 d(Ax,Ay)=\frac{d(x,y)}{\norm{Ax}\norm{Ay}}.
\end{equation}
Thus growth of the two vector norms leads to contraction of
their projective separation. In expressions such as $d(Ax,Ay)$,
we suppress the brackets denoting projective classes.

Furstenberg's formula identifies the stationary mean logarithmic
growth with the top Lyapunov exponent. For the forward and
inverse laws here, this exponent is positive by the same theorem
used in Corollary~\ref{prop:furstenberg-used}
\cite{Furstenberg1963}:
\begin{equation}\label{eq:stationary-log-positive}
 \int_{\RP}\int_{\mathbb R}
  \log\frac{\norm{A_{E'}(v)^{\pm1}u}}{\norm{u}}
       \,d\mu(v)\,d\nu_{E'}^\pm([u])>0.
\end{equation}
\begin{samepage}
The logarithm is independent of the nonzero representative $u$.
We give the bounded two-dimensional argument for Le Page's
contraction estimate \cite[Section~II, Proposition~1]{LePage1989}.
Positive mean logarithmic growth gives a small negative-moment
bound, which controls the projective distance through
\eqref{eq:sl2distance}.

\begin{lemma}[Finite-block fractional contraction]
\label{lem:contraction}
There exists a compact interval $I\subset\operatorname{int}\Sigma$
with $E$ in its interior for which one can choose an integer
$L\geq1$ and constants $r\in(0,1)$ and $q_0\in(0,1)$ such that,
uniformly over $E_-\in I$ and projective points $x,y$,
\begin{equation}\label{eq:blockcontract}
 \EE d(M_{E_-,L}x,M_{E_-,L}y)^r\leq q_0d(x,y)^r.
\end{equation}
The interval and constants depend only on $\mu$ and $E$.
\end{lemma}
\end{samepage}

\begin{proof}
At energy $E$, uniqueness of the forward stationary law and the
positive average \eqref{eq:stationary-log-positive} give a block
length $L$ with positive expected logarithmic growth uniformly
over starting directions. Indeed, otherwise Ces\`aro averages
of the direction laws from suitable starting points would have
a stationary limit with nonpositive average growth.
Continuity of the fixed-length product then gives a compact
interval $I\subset\operatorname{int}\Sigma$, with $E$ in its
interior, and $a>0$ such that
\[
 \inf_{\substack{t\in I\\\norm{u}=1}}
       \EE\log\norm{M_{t,L}u}\geq a.
\]

For this fixed $L$, $\log\norm{M_{t,L}u}$ is uniformly bounded
for $t\in I$ and unit $u$. Taylor's formula at $r=0$ therefore gives
\[
 \sup_{\substack{t\in I\\\norm{u}=1}}
       \EE\norm{M_{t,L}u}^{-2r}
       \leq1-2ar+Cr^2=:q_0.
\]
Choose $r\in(0,1)$ small enough that $q_0\in(0,1)$.
For unit representatives $x,y$, the determinant identity
\eqref{eq:sl2distance} and Cauchy--Schwarz now give
\begin{align*}
 \EE d(M_{t,L}x,M_{t,L}y)^r
 &\leq d(x,y)^r
       \bigl[\EE\norm{M_{t,L}x}^{-2r}
             \EE\norm{M_{t,L}y}^{-2r}\bigr]^{1/2}\\
 &\leq q_0d(x,y)^r.\qedhere
\end{align*}
\end{proof}

Fix $L,r,q_0$ and the interval $I$ supplied by
Lemma~\ref{lem:contraction}. We only need qualitative continuity,
so we use one common modulus throughout the estimates below:
\[
 \omega(h)=\omega_{\mathrm{en}}(h)+\omega_{\mathrm{ph}}(h)
 \longrightarrow0\qquad(h\downarrow0),
\]
with the moduli of Lemma~\ref{lem:continuity} taken on $I$.
This controls both energy and phase variation. All subsequent
limits are uniform over $E_-<E_+$ in $I$.

Using the angular form \eqref{eq:projective-distance} of the
projective distance, we now turn Lyapunov contraction into an
occupation estimate for the two-energy process. This will control the
middle of the phase interval in \eqref{eq:laplace-errors}.
We first perturb the contraction estimate to allow two nearby
energies, and then sum it up to a stopping time. The sum
cancels except for the first and last $L$ terms, giving the
bound below.

\begin{samepage}
\begin{lemma}[A stopped occupation bound]\label{lem:occupation}
With $L,r,q_0,I$ as in Lemma~\ref{lem:contraction}, for every
deterministic initial state and bounded stopping time $\tau$,
\begin{equation}\label{eq:stopped}
 (1-q_0)\EE\sum_{j<\tau}|\sin D_j|^r
       \leq L+C\delta^r\EE\tau.
\end{equation}
The constant $C$ may depend on the potential law $\mu$ and the
fixed reference energy $E$. It is uniform over $E_-<E_+$ in $I$,
the initial state and the stopping time.
\end{lemma}
\end{samepage}

\begin{proof}
To apply Lemma~\ref{lem:contraction}, both phases must evolve at
the same energy. Starting from $(\theta,\phi)$, let $\widehat\phi_k$ evolve at
$E_-$ from $\widehat\phi_0=\phi$, with the same potentials as
$\theta_k$ and $\phi_k$. Recall that $|v_j|\leq M$.
By \eqref{eq:derivatives},
$\partial_\theta F_t(v,\theta)\leq\norm{A_t(v)^{-1}}^2
\leq(t-v)^2+2$, so for $t\in I$ and $|v|\leq M$ we may take
the explicit common Lipschitz constant
\[
 B=2+\left(M+\max_{t\in I}|t|\right)^2.
\]
Combining this angle bound with the energy bound
\eqref{eq:energyclose} gives
\[
 |\phi_{k+1}-\widehat\phi_{k+1}|
 \leq B|\phi_k-\widehat\phi_k|+\delta.
\]
The initial difference is zero, so iteration gives
$|\phi_L-\widehat\phi_L|\leq\delta\sum_{k=0}^{L-1}B^k\leq C\delta$.
Since $\bigl||\sin a|^r-|\sin b|^r\bigr|\leq|a-b|^r$,
Lemma~\ref{lem:contraction}, together with the angular identity
\eqref{eq:projective-distance}, yields
\[
 \begin{aligned}
 \EE|\sin D_L|^r
 &\leq\EE|\sin(\theta_L-\widehat\phi_L)|^r
       +C\delta^r\\
 &\leq q_0|\sin(\theta-\phi)|^r+C\delta^r.
 \end{aligned}
\]

This estimate is uniform in the initial phases. Applying it
conditionally at time $j$ therefore gives
\begin{equation}\label{eq:two-energy-block}
 \EE[|\sin D_{j+L}|^r\mid\FF_j]
 \leq q_0|\sin D_j|^r+C\delta^r.
\end{equation}
We now sum this block estimate up to $\tau$.
Multiply by $\1_{\{j<\tau\}}$, sum over $j$ and take expectations to obtain
\[
 \EE\sum_{j<\tau}\bigl(|\sin D_{j+L}|^r-|\sin D_j|^r\bigr)
 \leq-(1-q_0)\EE\sum_{j<\tau}|\sin D_j|^r
       +C\delta^r\EE\tau.
\]
The sum on the left is bounded below by $-L$, because pathwise
\[
 \begin{aligned}
 \sum_{j<\tau}\bigl(|\sin D_{j+L}|^r-|\sin D_j|^r\bigr)
 &=\sum_{j=\tau}^{\tau+L-1}|\sin D_j|^r
       -\sum_{j=0}^{L-1}|\sin D_j|^r\\
 &\geq-L.
 \end{aligned}
\]
This proves \eqref{eq:stopped}.
\end{proof}

Estimate~\eqref{eq:stopped} controls the accumulated projective
separation in expectation up to a stopping time, without requiring the
initial phases to be close. At spatial block boundaries, we use
the corresponding bound at each deterministic time to confine the
phase gap near multiples of $\pi$.

\begin{samepage}
\begin{lemma}[Unstopped pair contraction]\label{lem:unstopped}
Uniformly over initial states and the fixed energy neighborhood,
\begin{equation}\label{eq:unstopped-pair}
 \EE|\sin D_j|^r\leq q_0^{\lfloor j/L\rfloor}+C\delta^r,
 \qquad j\geq0.
\end{equation}
The same bound holds after averaging any initial law independent of the
future potentials.
\end{lemma}
\end{samepage}

\begin{proof}
Write $j=kL+s$, with $0\leq s<L$. Iterating
\eqref{eq:two-energy-block} from time $s$ gives
\[
 \EE|\sin D_j|^r
 \leq q_0^k\EE|\sin D_s|^r+\frac{C\delta^r}{1-q_0}
 \leq q_0^k+\frac{C\delta^r}{1-q_0}.
\]
The bound is uniform in the initial state, so it may also be
integrated against an independent initial law.
\end{proof}

\subsection{Mean hitting times and the Laplace error}
\label{sec:mean-occupation}
The drift of $W_j$ gives useful information about $T$ once its
value at exit is controlled. The elementary crossing bound below
makes $W_T$ close to $\pi$. Stopping the drift will then give
the mean hitting time and show that little waiting time remains
on the IDS scale when $D_j$ is near $\pi$. This supplements projective contraction,
which controls the middle of the phase interval but treats gaps
near $0$ and $\pi$ alike.

The monotonicity in \eqref{eq:derivatives}, the shift identity
\eqref{eq:lift} and the energy bound \eqref{eq:energyclose} give,
for every integer $m$,
\[
 \begin{aligned}
 D_j\geq m\pi&\quad\Longrightarrow\quad D_{j+1}\geq m\pi,\\
 D_j<m\pi&\quad\Longrightarrow\quad D_{j+1}<m\pi+\delta.
 \end{aligned}
\]
Thus $D_j\geq0$, the count $K_j$ is nondecreasing, and at every
finite crossing time,
\begin{equation}\label{eq:overshoot}
 m\pi\leq D_{T_m}\leq m\pi+\delta.
\end{equation}
For $\delta<\pi$, each step crosses at most one new level.
These statements hold pathwise; $D_j$ may still move in either
direction within a period.

The Laplace argument requires two bounds on the mean hitting time.
The bound on $\Delta N\,\EE T$ justifies stopping and replacing
the geometric discount by an exponential. The bound near the
upper end controls the remaining contribution to the
accumulated error.

\begin{samepage}
\begin{lemma}[Mean exit times]\label{lem:mean}
For $0\leq D_0<\pi$, the first crossing time $T$ is integrable and
\begin{equation}\label{eq:mean-identity}
 \pi\Delta N\,\EE T=\EE W_T-W_0,
 \qquad
 \bigl|\pi\Delta N\,\EE T-(\pi-W_0)\bigr|
       \leq\omega(\delta).
\end{equation}
In particular, $\Delta N\,\EE T\leq C$ uniformly over all initial
phases. If $D_0\geq\pi-h$ with $0<h<\pi/2$, then
\begin{equation}\label{eq:uppermean}
 \pi\Delta N\,\EE T\leq\omega(h)+2\omega(\delta).
\end{equation}
\end{lemma}
\end{samepage}

\begin{proof}
Before $T$, the phase $\theta_j$ lies between $\phi_j$ and
$\phi_j+\pi$. At $T<\infty$, \eqref{eq:overshoot} places
$\theta_T-\pi$ between $\phi_T$ and $\phi_T+\delta$.
Monotonicity, the shift identity \eqref{eq:phase-function} and
the continuity moduli of $\mathcal H_t$ therefore give
\begin{equation}\label{eq:Wbounds}
 \begin{aligned}
 -\omega(\delta)\leq W_j&\leq\pi+\omega(\delta)
       &&(j<T),\\
 |W_T-\pi|&\leq\omega(\delta)
       &&(T<\infty).
 \end{aligned}
\end{equation}
The drift identity makes $W_j-\pi\Delta N\,j$ a martingale.
Stopping at $T\wedge m$ gives
\[
 \pi\Delta N\,\EE(T\wedge m)=\EE W_{T\wedge m}-W_0.
\]
By \eqref{eq:Wbounds}, the right side is bounded uniformly
in $m$. Since $\Delta N>0$, monotone
convergence gives $\EE T<\infty$. The stopped variables
$W_{T\wedge m}$ are uniformly bounded, so dominated convergence gives
\eqref{eq:mean-identity} and its error estimate.

Finally, when $D_0\geq\pi-h$, we have
$0\leq\phi_0+\pi-\theta_0\leq h$. By \eqref{eq:phase-function},
\[
 \begin{aligned}
 \pi-W_0
 &=\mathcal H_{E_-}(\phi_0+\pi)-\mathcal H_{E_-}(\theta_0)\\
 &\quad+\mathcal H_{E_+}(\phi_0)-\mathcal H_{E_-}(\phi_0)\\
 &\leq\omega(h)+\omega(\delta).
 \end{aligned}
\]
This proves \eqref{eq:uppermean}.
\end{proof}

Controlling the accumulated error in \eqref{eq:laplace-errors}
requires both estimates: Lyapunov contraction controls the middle
of the interval, and the drift of $W$ controls the time spent near
the upper end.

\begin{samepage}
\begin{lemma}[Occupation before the first crossing]\label{lem:killed}
Uniformly in $\theta_0$ and $0\leq D_0<\pi$,
\begin{equation}\label{eq:gapoccupation}
 \pi \Delta N\EE\sum_{j<T}|W_j|\longrightarrow0
 \quad(\delta\downarrow0).
\end{equation}
\end{lemma}
\end{samepage}

\begin{proof}
We split the sum into three regions: near $0$, in the middle, and
near $\pi$. The function $W_j$ is small in the first region; the time
spent in each of the other two is small on the IDS scale.
The constant $C$ below is uniform in the energies, initial state
and $h$.
By Lemma~\ref{lem:mean}, we have
\begin{equation}\label{eq:scaled-mean-budget}
 \pi \Delta N\EE T\leq C.
\end{equation}
Fix $0<h<\pi/4$.

\smallskip\noindent
\emph{1. Near $0$: $W_j$ is small.}
If $D_j\leq h$, then
$|W_j|\leq\omega(h)+\omega(\delta)$.
By \eqref{eq:scaled-mean-budget}, this region contributes at most
$C[\omega(h)+\omega(\delta)]$
to the left side of \eqref{eq:gapoccupation}.

\smallskip\noindent
\emph{2. In the middle: contraction controls occupation.}
Apply \eqref{eq:stopped} to $T\wedge m$ and let $m\to\infty$.
Nonnegativity and monotone convergence give
\begin{equation}\label{eq:fkilled}
 \pi \Delta N\EE\sum_{j<T}|\sin D_j|^r
 \leq\frac{L\pi \Delta N
                +C\delta^r(\pi \Delta N\EE T)}{1-q_0}
 \leq C(\Delta N+\delta^r).
\end{equation}
On $h<D_j<\pi-h$, we have $|\sin D_j|\geq\sin h$,
while \eqref{eq:Wbounds} gives
$|W_j|\leq\pi+\omega(\delta)$ for $j<T$.
Thus the middle region contributes at most
$C(\Delta N+\delta^r)/(\sin h)^r$.
For fixed $h$, this tends to zero. The error $C\delta^r$ in
\eqref{eq:two-energy-block} comes from comparing the phase recursions
at $E_-$ and $E_+$. Up to time $T$, this error accumulates to at
most $C\delta^r\EE T$. After multiplication by $\Delta N$,
\eqref{eq:scaled-mean-budget} gives
$C\delta^r\Delta N\,\EE T=O(\delta^r)$.
Therefore we do not need $\delta^r/\Delta N\to0$:
continuity of $N$ suffices here, without a quantitative
H\"older bound on its increments.

\smallskip\noindent
\emph{3. Near $\pi$: use the mean exit-time bound.}
Let
\[
 \sigma_h=\inf\{j\geq0:D_j\geq\pi-h\}.
\]
We have $\sigma_h\leq T$. On $\{\sigma_h<T\}$, the gap at
$\sigma_h$ lies in $[\pi-h,\pi)$, and the unused potentials
are iid and independent of $\FF_{\sigma_h}$. We can therefore
apply \eqref{eq:uppermean} to the phase pair restarted at $\sigma_h$.
On $\{\sigma_h=T\}$ the remaining time is zero, so in both cases
\begin{equation}\label{eq:remaining}
 \pi\Delta N\EE[T-\sigma_h\mid\FF_{\sigma_h}]
 \leq\omega(h)+2\omega(\delta).
\end{equation}
Every visit to $[\pi-h,\pi)$ before $T$ occurs between
$\sigma_h$ and $T$. Since \eqref{eq:Wbounds} bounds $|W_j|$
uniformly before $T$, taking expectations in \eqref{eq:remaining} gives
\[
 \begin{aligned}
 \pi\Delta N\EE\sum_{j<T}|W_j|
                           \1_{\{D_j\geq\pi-h\}}
 &\leq C\pi\Delta N\EE(T-\sigma_h)\\
 &\leq C[\omega(h)+\omega(\delta)].
 \end{aligned}
\]

\smallskip\noindent
\emph{Combining the three regions,} we obtain
\begin{equation}\label{eq:killed-error-ledger}
 \pi \Delta N\EE\sum_{j<T}|W_j|
 \leq C\Bigl[\omega(h)
       +\frac{\Delta N+\delta^r}{(\sin h)^r}
       +\omega(\delta)\Bigr].
\end{equation}
The constant $C$ is independent of $h$. First let $\delta\to0$
at fixed $h$, then let $h\to0$. This proves the uniform limit.
\end{proof}

\subsection{The Laplace transform and the exponential law}
\label{sec:laplace-proof}
\begin{proof}[Proof of Theorem~\ref{thm:exit}]
We prove the Laplace profile \eqref{eq:affineprofile} first.
Fix $s>0$. We let $\delta=E_+-E_-\downarrow0$, allowing
$E_\pm\in I$ and the initial phases to vary. Continuity of $N$
on the compact interval $I$ gives $\Delta N\to0$ uniformly.
For each energy pair, set $\rho=(1+s\Delta N)^{-1}$. This choice makes
$1-\rho=s\Delta N\rho$, so discounting the constant part of
$1+(s/\pi)W$ cancels the contribution of the exact drift.
More explicitly, \eqref{eq:drift} gives
\[
 \begin{aligned}
 &\EE\!\left[
   \rho^{j+1}\left(1+\frac{s}{\pi}W_{j+1}\right)
   -\rho^j\left(1+\frac{s}{\pi}W_j\right)\,
   \middle|\,\FF_j\right]\\
 &\qquad=\rho^j\left[
       (\rho-1)\left(1+\frac{s}{\pi}W_j\right)
                         +\rho s\Delta N\right]\\
 &\qquad=-\frac{s^2}{\pi}\Delta N\,\rho^{j+1}W_j.
 \end{aligned}
\]
With the energy pair fixed, stop at $T\wedge m$ and let
$m\to\infty$. Lemma~\ref{lem:mean}
and the uniform bound on the stopped values of $W$ justify dominated
convergence, giving
\begin{equation}\label{eq:geometric-exit-identity}
 \EE\!\left[\rho^T\left(1+\frac{s}{\pi}W_T\right)\right]
       -1-\frac{s}{\pi}W_0
 =-\frac{s^2}{\pi}\Delta N\,
       \EE\sum_{j<T}\rho^{j+1}W_j.
\end{equation}
\[
 \begin{aligned}
 &\left|(1+s)\EE\rho^T-1-\frac{s}{\pi}W_0\right|\\
 &\quad\overset{\eqref{eq:geometric-exit-identity}}{\leq}
       \frac{s}{\pi}\EE|W_T-\pi|
       +\frac{s^2}{\pi}\Delta N\,\EE\sum_{j<T}|W_j|\\
 &\quad\overset{\eqref{eq:Wbounds}}{\leq}
       \frac{s}{\pi}\omega(\delta)
       +\frac{s^2}{\pi}\Delta N\,\EE\sum_{j<T}|W_j|\\
 &\quad\xrightarrow[\eqref{eq:gapoccupation}]{\delta\downarrow0}0.
 \end{aligned}
\]
To replace the geometric discount by an exponential, use
$\log(1+s\Delta N)=s\Delta N+O((\Delta N)^2)$ and
$|e^{-u}-e^{-v}|\leq|u-v|$ for $u,v\geq0$. Then
\[
 \begin{aligned}
 \EE\bigl|\rho^T-e^{-s\Delta N T}\bigr|
 &=\EE\bigl|(1+s\Delta N)^{-T}-e^{-s\Delta N T}\bigr|\\
 &\leq C(\Delta N)^2\EE T
 \leq C\Delta N\xrightarrow{\delta\downarrow0}0,
 \end{aligned}
\]
by the mean bound $\Delta N\,\EE T\leq C$.
This proves \eqref{eq:affineprofile}.

For the exponential law, assume $0\leq D_0\leq h_\delta$, where
$h_\delta\to0$ as $\delta\downarrow0$. Then
\[
 |W_0|\leq\omega(h_\delta)+\omega(\delta)
 \xrightarrow{\delta\downarrow0}0.
\]
Together with \eqref{eq:affineprofile}, this gives
\[
 \EE e^{-s\Delta N T}
 \xrightarrow{\delta\downarrow0}\frac1{1+s},
\]
the Laplace transform of $\operatorname{Exp}(1)$.
The error bounds hold for every allowed energy pair and initial
state. Applying the Laplace continuity theorem
\cite[Chapter~XIII, Section~1]{Feller1971} along any such sequence
with $\delta\downarrow0$ proves the uniform convergence in
\eqref{eq:exp}.
\end{proof}

\section{Renewal counts and phase gaps modulo \texorpdfstring{$\pi$}{pi}}\label{sec:renewal}
We use Theorem~\ref{thm:exit} to obtain independent exponential
limits for successive waiting times and Poisson limits for crossing
counts. We then control
the phase gaps modulo $\pi$ and terminal phases needed to recover
eigenvalue counts in Section~\ref{sec:pointprocess}.
Throughout, $I$ is the compact energy interval fixed in
Section~\ref{sec:exit}.

\subsection{Restarts and Poisson counts}
For $E_-<E_+$ in $I$, write $\Delta N=N(E_+)-N(E_-)$ and
take $\delta=E_+-E_-<\pi$, so each step crosses at most one new level.
Write
\[
 R_j=D_j-\pi K_j\in[0,\pi)
\]
for the phase gap modulo $\pi$, represented in $[0,\pi)$.

\begin{samepage}
\begin{corollary}[Successive crossing times]
\label{cor:restart}
Let $E_-<E_+$ lie in $I$.
As $\delta\to0$, start the coupled phases of \eqref{eq:pair}
with $0\leq D_0\leq h_\delta\to0$.
Recall that the $m$th crossing time is
\[
 T_m=\inf\{j\geq0:D_j\geq m\pi\},\qquad
 D_j=\theta_j-\phi_j,\qquad m\geq1,
\]
and set $T_0=0$.
Every $T_m$ is finite almost surely. For each fixed $M$,
\[
 \Delta N\bigl(T_1-T_0,\ldots,T_M-T_{M-1}\bigr)
            \Longrightarrow(X_1,\ldots,X_M),
\]
where the $X_i$ are independent $\operatorname{Exp}(1)$ variables.
The convergence is uniform over the allowed energy pairs, $\theta_0$,
and the allowed initial gap.
\end{corollary}
\end{samepage}

\begin{proof}
Whenever $T_{m-1}<\infty$, add $\pi(m-1)$ to the second lift.
By \eqref{eq:lift}, this shift preserves the recursion \eqref{eq:pair}, while
the gap becomes $R_{T_{m-1}}$. Its next crossing time is therefore
$T_m-T_{m-1}$. The unused potentials are iid and independent of
$\FF_{T_{m-1}}$.

Initially $R_{T_0}=D_0\leq h_\delta$; at every later crossing,
the overshoot bound \eqref{eq:overshoot} gives
$0\leq R_{T_{m-1}}\leq\delta$.
The uniform mean bound of Lemma~\ref{lem:mean} therefore proves
inductively that all $T_m$ are integrable and finite. Theorem~\ref{thm:exit}
then gives, for each fixed $s>0$,
\begin{equation}\label{eq:conditional-restart-laplace}
 \EE\!\left[e^{-s\Delta N(T_m-T_{m-1})}
              \mid\FF_{T_{m-1}}\right]=\frac1{1+s}+o(1)
 \qquad(\delta\downarrow0).
\end{equation}
The absolute error has a deterministic bound tending to zero,
uniformly over the entering phase, the allowed energies and the
crossing number.

Fix $q_1,\ldots,q_M\geq0$ and condition successively at
$T_{M-1},\ldots,T_0$. All exponential factors are bounded by one,
so the finitely many uniform errors in
\eqref{eq:conditional-restart-laplace} tend to zero. Hence
\[
 \EE\exp\!\left[-\Delta N
                         \sum_{m=1}^Mq_m(T_m-T_{m-1})\right]
             \longrightarrow\prod_{m=1}^M(1+q_m)^{-1}.
\]
Joint Laplace-transform convergence gives the asserted independent
exponential limit.
\end{proof}

To pass from waiting times to counts, we use the fact that the
count has reached $k$ by a given time precisely when the $k$th
crossing has occurred. The limiting arrival times are sums of
the exponentials in Corollary~\ref{cor:restart}.
For $k\geq1$, let $\Gamma_k$ denote the distribution function
of the sum of $k$ independent $\operatorname{Exp}(1)$ variables.
These sums have the standard unit-rate Gamma laws
\cite[Theorem~7.2 and Exercise~7.3]{LastPenrose2018}, whose distribution
functions are
\begin{equation}\label{eq:gamma-arrival-cdf}
 \Gamma_k(t)=1-e^{-t}\sum_{j=0}^{k-1}\frac{t^j}{j!},\qquad t\geq0.
\end{equation}
In particular,
$\Gamma_1(t)=1-e^{-t}$ and $\Gamma_2(t)=1-e^{-t}(1+t)$.
The elementary bounds needed later are
\begin{equation}\label{eq:gamma-small-time}
 0\leq t-\Gamma_1(t)\leq\frac{t^2}{2},\qquad
 0\leq\Gamma_2(t)\leq\frac{t^2}{2}\qquad(t\geq0).
\end{equation}

To obtain local spectral statistics, we now let the energy window
vary with the volume. The arrival-time limits give a Poisson law
for the count on spatial intervals of length of order $1/\Delta N$.

\begin{samepage}
\begin{corollary}[Poisson counts for one energy window]\label{cor:counts}
Let $E_{-,n}<E_{+,n}$ lie in $I$, with
$\delta_n=E_{+,n}-E_{-,n}\to0$, and write
$\Delta N_n=N(E_{+,n})-N(E_{-,n})$. Start the coupled phases with
$0\leq D_0\leq h_n\to0$. If the deterministic integers $\ell_n\geq0$
satisfy $\Delta N_n\ell_n\to t\geq0$, then
\[
 K_{\ell_n}\ \Longrightarrow\ \operatorname{Poisson}(t).
\]
More precisely, for every fixed $k\geq1$,
\begin{equation}\label{eq:uniform-arrival-kernel}
 \sup_{\theta_0\in\mathbb R,\,0\leq D_0\leq h_n}
 \left|\PP(T_k\leq\ell_n)-\Gamma_k(t)\right|
 \xrightarrow[n\to\infty]{\delta_n\to0}0.
\end{equation}
In particular, if $n\Delta N_n\to w>0$, then
$K_{\lfloor ns\rfloor}\Longrightarrow\operatorname{Poisson}(ws)$
for each fixed $s\geq0$.
\end{corollary}
\end{samepage}

Uniformity in $\theta_0$ allows us to apply this result at the
spatial block boundaries in Section~\ref{sec:pointprocess}, where
the starting phase depends on the past.

\begin{proof}
By Corollary~\ref{cor:restart}, the limiting distribution function
of $\Delta N_nT_k$ is $\Gamma_k$. Its continuity
permits the moving threshold $\Delta N_n\ell_n\to t$, and the
uniformity in the initial state gives \eqref{eq:uniform-arrival-kernel}.
Since $\{K_{\ell_n}\geq k\}=\{T_k\leq\ell_n\}$, adjacent tails give
\[
 \PP(K_{\ell_n}=k)\longrightarrow
 \Gamma_k(t)-\Gamma_{k+1}(t)=e^{-t}\frac{t^k}{k!},\qquad k\geq0,
\]
where $\Gamma_0(t)=1$. These limits sum to one, proving the count law.
\end{proof}

To control the phase gap modulo $\pi$, we first estimate the probability of
a crossing in a short future interval.

\begin{samepage}
\begin{lemma}[Crossings in a short future window]\label{lem:short-window-counts}
Let the energy pairs be as in Corollary~\ref{cor:counts}, with
$n\Delta N_n\to w>0$, and fix any sequence $h_n\geq0$ tending
to zero. Then, for every fixed $t\geq0$,
\[
 \lim_{\varepsilon\downarrow0}\limsup_{n\to\infty}
 \PP\bigl(K_{\lfloor n(t+\varepsilon)\rfloor}
                  >K_{\lfloor nt\rfloor}\bigr)=0,
\]
uniformly over $\theta_0\in\mathbb R$ and $0\leq D_0\leq h_n$.
\end{lemma}
\end{samepage}

\begin{proof}
Fix an integer $M\geq1$ to truncate the number of crossings.
A crossing in this window is either among the first $M$ crossings
or implies $T_{M+1}\leq\lfloor n(t+\varepsilon)\rfloor$.
Thus \eqref{eq:uniform-arrival-kernel} gives
\[
 \begin{aligned}
 &\limsup_{n\to\infty}
 \PP\bigl(K_{\lfloor n(t+\varepsilon)\rfloor}>K_{\lfloor nt\rfloor}\bigr)\\
 &\qquad\leq\Gamma_{M+1}\bigl(w(t+\varepsilon)\bigr)
       +\sum_{k=1}^M\bigl[\Gamma_k\bigl(w(t+\varepsilon)\bigr)
                                    -\Gamma_k(wt)\bigr].
 \end{aligned}
\]
First let $\varepsilon\downarrow0$. Continuity of the $\Gamma_k$
leaves only $\Gamma_{M+1}(wt)$, which tends to zero as $M\to\infty$.
\end{proof}

\subsection{Phase gaps modulo \texorpdfstring{$\pi$}{pi} and the terminal phase}

At a crossing time $T_m$, \eqref{eq:overshoot} gives
$R_{T_m}\leq\delta$. We next show that $R_j$ tends to zero in
probability at times $j=\lfloor nt\rfloor$, for each fixed $t>0$.

\begin{samepage}
\begin{lemma}[Phase gaps modulo $\pi$ at deterministic times]\label{lem:residual-pinning}
Let $E_{-,n}<E_{+,n}$ tend to $E$, with $n\Delta N_n\to w>0$.
Start the coupled phases \eqref{eq:pair} from equal angles.
Then, for every $t>0$,
\[
 R_{\lfloor nt\rfloor}\longrightarrow0
 \quad\hbox{in probability}.
\]
\end{lemma}
\end{samepage}

\begin{proof}
Fix $t>0$, set $j_n=\lfloor nt\rfloor$, and fix a threshold $h_0>0$.
For $0<h<\min(h_0,\pi/4)$, \eqref{eq:unstopped-pair} gives
\begin{equation}\label{eq:residual-middle-small}
 \PP(h\leq R_{j_n}\leq\pi-h)\longrightarrow0.
\end{equation}
To exclude $R_{j_n}$ near $\pi$, fix $\eta>0$ and put
$\ell_n=\lfloor n(t+\eta)\rfloor-j_n$, so that
$\Delta N_n\ell_n\to w\eta$. Conditional on $\FF_{j_n}$,
the time to the next crossing is the first exit time restarted
from $R_{j_n}$. The upper-end mean bound \eqref{eq:uppermean}
and Markov's inequality give
\[
 \PP(R_{j_n}\geq\pi-h,\ K_{j_n+\ell_n}=K_{j_n})
 \leq\frac{\omega(h)+2\omega(\delta_n)}
              {\pi\Delta N_n\ell_n}.
\]
If $R_{j_n}>h_0$, either $R_{j_n}$ lies in the middle region,
a crossing occurs during the future window, or the event in the
last display occurs. Hence
\begin{equation}\label{eq:short-window-arrival-budget}
 \limsup_{n\to\infty}\PP(R_{j_n}>h_0)
 \leq\limsup_{n\to\infty}\PP(K_{j_n+\ell_n}>K_{j_n})
       +\frac{\omega(h)}{\pi w\eta}.
\end{equation}
First let $h\downarrow0$ at fixed $\eta$, then let
$\eta\downarrow0$. Lemma~\ref{lem:short-window-counts} makes the
right side vanish. Since $h_0>0$ was arbitrary, $R_{j_n}$
converges to zero in probability.
\end{proof}

Converting phase counts into eigenvalue counts requires one more
estimate. A small phase gap modulo $\pi$ can still produce a rounding
error if the terminal phase is close to the Dirichlet direction
$\pi/2$. We exclude this possibility using convergence to the
non-atomic forward stationary law. The energy endpoints will
vary with $n$, so we state this convergence for moving energies.

\pagebreak[0]
\begin{samepage}
For angular distance modulo $\pi$, write
\[
 \operatorname{dist}_{\pi}(x,y)
   =\min_{k\in\mathbb Z}|x-y-k\pi|.
\]

\begin{lemma}[Terminal phases at varying energies]
\label{lem:forward-stationary-limit}
Let $t_n$ be deterministic energies tending to $E$, and let
$\theta_n(t_n)$ be the Dirichlet phase defined by
\eqref{eq:dirichlet-phase}.
Then
\[
 \theta_n(t_n)\bmod\pi\ \Longrightarrow\ \nu_E^+.
\]
In particular, for every projective angle $\beta$,
\begin{equation}\label{eq:moving-phase-no-atom}
 \lim_{h\downarrow0}\limsup_{n\to\infty}
 \PP\bigl(\operatorname{dist}_{\pi}(\theta_n(t_n),\beta)\leq h\bigr)=0.
\end{equation}
\end{lemma}
\end{samepage}

\begin{proof}
At energy $t_n$, couple the Dirichlet phase with a phase
$\widehat\theta_j^{(n)}$ whose initial projective law is
$\nu_{t_n}^+$, independently of the future potentials.
Iterating the uniform contraction estimate \eqref{eq:blockcontract}
of Lemma~\ref{lem:contraction} gives
\[
 \widehat\theta_n^{(n)}\bmod\pi\sim\nu_{t_n}^+,\qquad
 \EE d(\theta_n(t_n),\widehat\theta_n^{(n)})^r
       \leq Cq_0^{\lfloor n/L\rfloor}\longrightarrow0.
\]
Lemma~\ref{lem:continuity} gives $\nu_{t_n}^+\Longrightarrow\nu_E^+$.
The coupled phases are asymptotically equal on the projective
line, which proves the first assertion. The closed-set inequality
in the Portmanteau theorem \cite[Theorem~2.1(b)]{Billingsley1971}
gives, for fixed $h>0$,
\[
 \begin{aligned}
 &\limsup_{n\to\infty}
   \PP\bigl(\operatorname{dist}_{\pi}(\theta_n(t_n),\beta)\leq h\bigr)\\
 &\qquad\leq\nu_E^+\bigl(\{x\in\RP:
                    \operatorname{dist}_{\pi}(x,\beta)\leq h\}\bigr)
   \xrightarrow{h\downarrow0}\nu_E^+(\{\beta\})=0,
 \end{aligned}
\]
where the last equality uses non-atomicity. This proves
\eqref{eq:moving-phase-no-atom}.
\end{proof}

\section{Asymptotic independence and the spectral point process}
\label{sec:pointprocess}

Corollary~\ref{cor:counts} gives Poisson limits for individual energy
windows. We now prove asymptotic independence across windows and
convert phase counts to eigenvalue counts to obtain
Theorem~\ref{thm:poisson}.

\subsection{Asymptotic independence across energy intervals}

Fix the number $s\geq1$ of energy intervals. For each $n$, take
their $s+1$ ordered endpoints $E_{0,n}<\cdots<E_{s,n}$.
The phases at these energies all start from zero and use the same
potentials, with
$\FF_j=\sigma(v_0,\ldots,v_{j-1})$.
For $1\leq i\leq s$ and $0\leq j\leq n$, define the corresponding
phase gaps and counts by
\[
 D_i(j)=\theta_j(E_{i-1,n})-\theta_j(E_{i,n}),\qquad
 K_i(j)=\lfloor D_i(j)/\pi\rfloor,\qquad
 R_i(j)=D_i(j)-\pi K_i(j).
\]
For each energy interval $(E_{i-1,n},E_{i,n}]$, write its IDS increment as
\begin{equation}\label{eq:window-gaps}
 \Delta N_{i,n}=N(E_{i,n})-N(E_{i-1,n}).
\end{equation}
For the full energy interval $(E_{0,n},E_{s,n}]$, define
\[
 \begin{aligned}
 D_*(j)&=\theta_j(E_{0,n})-\theta_j(E_{s,n})
                         =\sum_{i=1}^s D_i(j),\\
 K_*(j)&=\left\lfloor\frac{D_*(j)}\pi\right\rfloor.
 \end{aligned}
\]

For each adjacent window, Corollary~\ref{cor:counts} already
gives a Poisson limit determined by its IDS mass. It remains to
prove independence across the windows.
We use the full-interval count $K_*$ to control multiple crossings
across the windows.

\begin{samepage}
\begin{proposition}[Asymptotic independence of crossing counts]\label{prop:joint-crossing-counts}
Let $E\in\operatorname{int}\Sigma$. Suppose $E_{i,n}\to E$ for $0\leq i\leq s$, and
$n\Delta N_{i,n}\to w_i>0$ for $1\leq i\leq s$. Then the crossing
counts in the disjoint energy intervals $(E_{i-1,n},E_{i,n}]$,
$1\leq i\leq s$, are asymptotically independent:
\[
 (K_1(n),\ldots,K_s(n))\Longrightarrow (Z_1,\ldots,Z_s),
\]
where the $Z_i$ are independent Poisson variables of means $w_i$.
\end{proposition}
\end{samepage}

\begin{proof}
The subdivision argument is related to the phase-counting proofs
in \cite[Section~3.2]{AD14} and \cite[Section~5.3]{KN17}.
Fix a large integer $m$ and divide the $n$ sites into $m$ blocks.
Blocks use disjoint potential variables, but their
entering phases depend on the past. Lemma~\ref{lem:residual-pinning} and the
uniform restart law make the conditional crossing probabilities
asymptotically deterministic.

\smallskip
\emph{Block starts and the full energy interval.}
The floor operation gives the exact identity
\begin{equation}\label{eq:joint-floor-identity}
 K_*(j)=\sum_{i=1}^s K_i(j)+\left\lfloor\frac{\sum_{i=1}^s R_i(j)}\pi\right\rfloor.
\end{equation}
Define the block endpoints by
\[
 j_b=\left\lfloor\frac{bn}{m}\right\rfloor,\qquad 0\leq b\leq m.
\]
With $s$ and $m$ fixed, Lemma~\ref{lem:residual-pinning},
together with $R_i(0)=0$, allows us to choose a common deterministic
sequence $h_n\downarrow0$ such that, for $1\leq b\leq m$, the events
\[
 G_{b,n}=\left\{\max_{1\leq i\leq s}R_i(j_{b-1})\leq h_n\right\}
\]
satisfy $\max_{1\leq b\leq m}\PP(G_{b,n}^{\complement})
\xrightarrow{n\to\infty}0$.
On $G_{b,n}$, the restarted adjacent gaps are at most $h_n$,
and the gap for the full interval is at most $sh_n$.

Let $w_* =\sum_{i=1}^s w_i$ and define the block increments by
\[
 \Delta K_i^{(b)}=K_i(j_b)-K_i(j_{b-1}),\qquad
 \Delta K_*^{(b)}=K_*(j_b)-K_*(j_{b-1}).
\]
On $G_{b,n}$, $\sum_{i=1}^s R_i(j_{b-1})\leq sh_n<\pi$ for
large $n$, so the initial floor in \eqref{eq:joint-floor-identity}
vanishes. Subtracting that identity at the two block endpoints gives
\[
 \sum_{i=1}^s\Delta K_i^{(b)}
 =\Delta K_*^{(b)}
  -\left\lfloor\frac{\sum_{i=1}^s R_i(j_b)}\pi\right\rfloor
 \leq\Delta K_*^{(b)}.
\]
The terminal floor is nonnegative because each $R_i(j_b)\geq0$.

\smallskip
\emph{Crossings within one spatial block.}
After removing the completed periods from $D_i$ at $j_{b-1}$,
$\Delta K_i^{(b)}\geq k$ means exactly that the restarted $T_k$
is at most $j_b-j_{b-1}$. Since
$\Delta N_{i,n}(j_b-j_{b-1})\to w_i/m$,
\eqref{eq:uniform-arrival-kernel}, applied conditionally on
$\FF_{j_{b-1}}$ and also to the two extreme energies,
gives a common deterministic $\varepsilon_n(m)\to0$ such that,
on $G_{b,n}$,
\begin{align*}
 \left|\PP(\Delta K_i^{(b)}\geq1\mid\FF_{j_{b-1}})
                    -(1-e^{-w_i/m})\right|&\leq\varepsilon_n(m),\\
 \PP(\Delta K_*^{(b)}\geq2\mid\FF_{j_{b-1}})
                    &\leq\frac{w_*^2}{2m^2}+\varepsilon_n(m).
\end{align*}
The error is uniform in the restart phases, in gaps at most
$sh_n$, and in the finitely many block and window indices.
The second bound uses
$\Gamma_2(w_*/m)\leq w_*^2/(2m^2)$ from
\eqref{eq:gamma-small-time}.

For fixed $t_1,\ldots,t_s\geq0$, define
$\Lambda=\sum_{i=1}^s w_i(1-e^{-t_i})$.
When $\sum_i\Delta K_i^{(b)}\leq1$, the exponential below equals
$1-\sum_i(1-e^{-t_i})\1_{\{\Delta K_i^{(b)}\geq1\}}$.
Otherwise the discrepancy is bounded by a constant depending only
on $s$, and this event requires $\Delta K_*^{(b)}\geq2$. Thus,
on $G_{b,n}$,
\[
 \begin{aligned}
 &\EE\!\left[e^{-\sum_{i=1}^s t_i\Delta K_i^{(b)}}\mid\FF_{j_{b-1}}\right]\\
 &\quad=1-\sum_{i=1}^s(1-e^{-t_i})
       \PP(\Delta K_i^{(b)}\geq1\mid\FF_{j_{b-1}})
       +O(m^{-2}+\varepsilon_n(m)).
 \end{aligned}
\]
Using $1-e^{-w_i/m}=w_i/m+O(m^{-2})$ yields
\begin{equation}\label{eq:joint-block-kernel}
 \left|\EE\!\left[e^{-\sum_{i=1}^s t_i\Delta K_i^{(b)}}
                    \mid\FF_{j_{b-1}}\right]-(1-\Lambda/m)\right|
 \leq C m^{-2}+C\varepsilon_n(m).
\end{equation}
Here $C$ depends only on the fixed windows and their limiting rates.

\smallskip
\emph{Combining the spatial blocks.}
Take $m>w_*$, so $1-\Lambda/m\in[0,1]$.
Take conditional expectations from the last block backwards, using
\eqref{eq:joint-block-kernel} on each $G_{b,n}$ and the bound $2$
on its complement. All earlier block factors are bounded by one, giving
\begin{equation}\label{eq:joint-backward-error}
 \left|\EE e^{-\sum_{i=1}^s t_iK_i(n)}-(1-\Lambda/m)^m\right|
 \leq C/m+Cm\varepsilon_n(m)
             +2\sum_{b=1}^m\PP(G_{b,n}^{\complement}).
\end{equation}
First let $n\to\infty$ at fixed $m$, and then let $m\to\infty$.
The joint Laplace transform then converges to
$\exp[-\sum_{i=1}^s w_i(1-e^{-t_i})]$.
The multivariate Laplace continuity theorem
\cite[Proposition~B.10]{LastPenrose2018} proves the asserted
independent Poisson limit.
\end{proof}

\subsection{From crossings to spectral counts}
Sturm oscillation leaves a possible error of one per energy
window. This can occur only when the terminal phase
$\theta_n(E_{i,n})$ lies within $R_i(n)$ of the Dirichlet direction
$\pi/2$ modulo $\pi$. Lemmas~\ref{lem:residual-pinning}
and~\ref{lem:forward-stationary-limit} show that the probability
of this event tends to zero.

\begin{samepage}
\begin{lemma}[From phase counts to eigenvalue counts]
\label{lem:spectral-count-conversion}
Let the endpoint arrays satisfy the hypotheses of
Proposition~\ref{prop:joint-crossing-counts}, and set
\[
 C_i(n)=\#\{j:\lambda_{j,n}\in(E_{i-1,n},E_{i,n}]\}.
\]
Then
\[
 \PP\bigl(C_i(n)=K_i(n)\text{ for all }1\leq i\leq s\bigr)
 \longrightarrow1.
\]
\end{lemma}
\end{samepage}

\begin{proof}
Since $D_i(n)=\pi K_i(n)+R_i(n)$, the Sturm ceiling identity gives
\begin{equation}\label{eq:exact-rounding-error}
 \begin{aligned}
 C_i(n)-K_i(n)
 &=\left\lceil\frac{\theta_n(E_{i,n})-\pi/2+R_i(n)}{\pi}\right\rceil
   -\left\lceil\frac{\theta_n(E_{i,n})-\pi/2}{\pi}\right\rceil\\
 &\in\{0,1\}.
 \end{aligned}
\end{equation}
If $C_i(n)-K_i(n)=1$, then
$\operatorname{dist}_{\pi}(\theta_n(E_{i,n}),\pi/2)\leq R_i(n)$.
Thus, for every $h>0$,
\[
 \begin{aligned}
 \PP(\exists i:C_i(n)\ne K_i(n))
 &\leq\sum_{i=1}^s\PP(R_i(n)>h)\\
 &\quad+\sum_{i=1}^s
 \PP\bigl(\operatorname{dist}_{\pi}(\theta_n(E_{i,n}),\pi/2)\leq h\bigr).
 \end{aligned}
\]
For fixed $h$, the first sum tends to zero by
Lemma~\ref{lem:residual-pinning}. For the second sum,
Lemma~\ref{lem:forward-stationary-limit} gives
\[
 \lim_{h\downarrow0}\limsup_{n\to\infty}
 \sum_{i=1}^s\PP\bigl(
   \operatorname{dist}_{\pi}(\theta_n(E_{i,n}),\pi/2)\leq h\bigr)=0.
 \qedhere
\]
\end{proof}

\subsection{Completion of the point-process limit}
The joint spectral count limits are now available for energy
windows with prescribed limiting IDS masses.
We use the standard Laplace-functional criterion for vague
convergence \cite[Proposition~3.19]{Resnick1987}. For the Poisson
process $\Pi$, the Laplace functional is
\begin{equation}\label{eq:poisson-laplace}
 \EE e^{-\int f\,d\Pi}
   =\exp\left[-\int_{\mathbb R}(1-e^{-f(x)})\,dx\right],
 \qquad 0\leq f\in C_c(\mathbb R);
\end{equation}
see \cite[Theorem~3.9]{LastPenrose2018}.
By a standard step-function approximation, it suffices to prove
joint Poisson limits for counts in finite partitions of bounded
intervals.

\begin{proof}[Proof of Theorem~\ref{thm:poisson}]
Fix a partition
\[
 -R=x_0<x_1<\cdots<x_s=R.
\]
For IDS unfolding, choose an open interval $J$ about $E$ with
$J\subset\operatorname{int}\Sigma$. Continuity and strict
increase of $N$ on $J$ imply that, for large $n$, the equations
\[
 N(E_{i,n})=N(E)+\frac{x_i}{n},\qquad 0\leq i\leq s,
\]
determine unique endpoints $E_{i,n}\in J$.
These endpoints tend to $E$, and their adjacent IDS masses satisfy
$n\Delta N_{i,n}=x_i-x_{i-1}$. By monotonicity,
\[
 \xi_{n,E}((x_{i-1},x_i])
 =\#\{j:\lambda_{j,n}\in(E_{i-1,n},E_{i,n}]\}.
\]
Energies outside $J$ do not contribute to $[-R,R]$ once $n$ is
large, since $N(E)$ is an interior point of $N(J)$.
Proposition~\ref{prop:joint-crossing-counts} and
Lemma~\ref{lem:spectral-count-conversion} therefore give the
required independent Poisson limits, proving the first conclusion.

For linear scaling, assume in addition that $N$ is differentiable
at $E$, with finite positive derivative $\rho(E):=N'(E)$. We then set
\[
 E_{i,n}=E+\frac{x_i}{n\rho(E)}.
\]
Differentiability at $E$ gives
$n[N(E_{i,n})-N(E)]\to x_i$, so the adjacent rates again tend
to $x_i-x_{i-1}$. The same two results give the partition
count limits for the linearly rescaled process, completing the proof.
\end{proof}

\begin{samepage}
\appendix
\section{Formalization and AI disclosure}\label{sec:formalization}
OpenAI's GPT-6 assisted with mathematical discussions, Lean
formalization and revisions to the exposition. The author directed
the work and takes responsibility for the mathematical claims and
the manuscript.

A Lean~4 formalization using mathlib is available in the
\href{https://github.com/hanyi162013-Yihan/Poisson_Statistic_1dAndersonBernoulli}
{accompanying repository}.
It formalizes the Poisson conclusions through the phase-crossing
argument developed above. The stationary-law input in
Corollary~\ref{prop:furstenberg-used} and the positive-growth
statement \eqref{eq:stationary-log-positive} are also proved in
the formalization for the transfer law of this model.
The paper cites Furstenberg's theorem for this input and Le Page's
theorem for projective contraction. The remaining crossing and
counting estimates are proved in
Sections~\ref{sec:exit}--\ref{sec:pointprocess}.

\end{samepage}

\clearpage
\printbibliography
\end{document}